\documentclass[12pt]{article}

\usepackage[utf8]{inputenc}
\usepackage[T1]{fontenc}

\usepackage{amssymb,amsthm}
\usepackage[hidelinks]{hyperref}
\usepackage{amsmath,graphicx,subfigure}
\usepackage{cite}
\usepackage{enumerate,enumitem}

\DeclareGraphicsExtensions{.pdf}

\newtheorem{proposition}{Proposition}
\newtheorem{lemma}{Lemma}
\newtheorem{theorem}{Theorem}
\newtheorem{corollary}{Corollary}

\newtheorem{ourclaim}{Claim}
\newenvironment{claimproof}{%
 \noindent%
 \textit{Proof of Claim.}%
}
{
\hfill $\triangle$%
\medskip
}

\title{The Density of Fan-Planar Graphs}

\author{Michael Kaufmann\\
\small Eberhard Karls Universit\"at T\"ubingen\\[-0.8ex]
\small\tt mk@informatik.uni-tuebingen.de\\
\and
Torsten Ueckerdt\\
\small Karlsruhe Institute of Technology\\[-0.8ex]
\small\tt torsten.ueckerdt@kit.edu}

\begin{document}

\maketitle

\begin{abstract}
 A topological drawing of a graph is fan-planar if for each edge $e$ the edges crossing $e$ form a star and no endpoint of $e$ is enclosed by $e$ and its crossing edges.
 A fan-planar graph is a graph admitting such a drawing. Equivalently, this can be formulated by three forbidden patterns, one of which is the configuration where $e$ is crossed by two independent edges and the other two where $e$ is crossed by two incident edges in a way that encloses some endpoint of $e$. 
 A topological drawing is simple if any two edges have at most one point in common.
 
 Fan-planar graphs are a new member in the ever-growing list of topological graphs defined by forbidden intersection patterns, such as planar graphs and their generalizations, Tur\'an-graphs and Conway's thrackle conjecture.
 Hence fan-planar graphs fall into an important field in combinatorial geometry with applications in various areas of discrete mathematics.
 As every $1$-planar graph is fan-planar and every fan-planar graph is $3$-quasiplanar, they also fit perfectly in a recent series of work on nearly-planar graphs from the area of graph drawing and combinatorial embeddings.

 In this paper we show that every fan-planar graph on $n$ vertices has at most $5n-10$ edges, even though a fan-planar drawing may have a quadratic number of crossings.
 Our bound, which is tight for every $n \geq 20$, indicates how nicely fan-planar graphs fit in the row with planar graphs ($3n-6$ edges) and $1$-planar graphs ($4n-8$ edges).
 With this, fan-planar graphs form the largest non-trivial class of topological graphs defined by forbidden patterns, for which the maximum number of edges on $n$ vertices is known exactly.

 We demonstrate that maximum fan-planar graphs carry a rich structure, which makes this class attractive for many algorithms commonly used in graph drawing.
 Finally, we discuss possible extensions and generalizations of these new concepts.
\end{abstract}

\section{Introduction}

Planarity of a graph is a well-studied concept in graph theory, computational geometry and graph drawing. The famous Euler formula characterizes for a certain embedding the relation between vertices, edges and faces, and many different algorithms e.g.~\cite{tutte1963draw} following different objectives have been developed to compute appropriate embeddings in the plane.

Because of the importance of the concepts, a series of generalizations have been developed in the past. Topological graphs and topological drawings respectively are being considered, i.e., the vertices are drawn as points in the plane and the edges drawn as Jordan curves between corresponding points without any other vertex as an interior point. In~\cite{DBLP:journals/siamdm/FoxPS13}, the authors state ``Finding the maximum number of edges in a topological graph with a forbidden crossing pattern is a fundamental problem in extremal topological graph theory'' together with 9 citations from a large group of authors. Most of the existent literature considers topological drawings that are \emph{simple}, i.e., where any two edges have at most one point in common.
In particular, two edges may not cross more than once and incident edges may not cross at all. We shall consider simple topological graphs only. Indeed, we shall argue in Section~\ref{sec:discussion} that if we drop this assumptions 
and allow non-homeomorphic parallel edges, then even $3$-vertex fan-planar graphs have arbitrarily many edges.

\begin{figure}[b]
 \centering
 \includegraphics{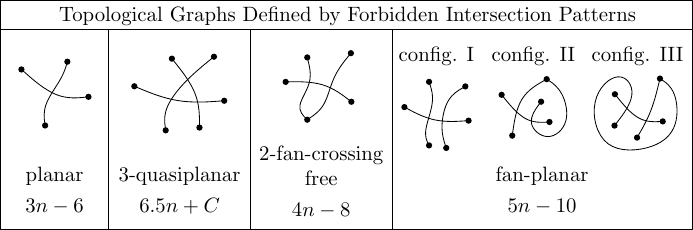}
 \caption{Topological graphs defined by forbidden patterns and the corresponding maximum number of edges in an $n$-vertex such graph.}
 \label{fig:forbidden-patterns}
\end{figure}

\paragraph{Related work.}
Most notably, there are $k$-planar graphs~\cite{DBLP:journals/combinatorica/PachT97} and $k$-quasi\-planar graphs~\cite{DBLP:journals/combinatorica/AgarwalAPPS97}.
A $k$-planar graph admits a topological drawing in which no edge is crossed more than $k$ times by other edges, while a $k$-quasiplanar graph admits a drawing in which no $k$ edges pairwise cross each other.

The topic of $k$-quasiplanar graphs is almost classical~\cite{DBLP:books/daglib/0017422}. 
A famous conjecture~\cite{DBLP:books/daglib/0017422} states that for constant $k$ the maximal number of edges in $k$-quasiplanar graphs is linear in the number of vertices. 
Note that $2$-quasiplanar graphs correspond to planar graphs. 
A first linear bound for $k=3$, i.e., $3$-quasiplanar graphs, has been shown in~\cite{DBLP:journals/combinatorica/AgarwalAPPS97} and subsequently improved in~\cite{DBLP:journals/combinatorica/PachT97}. 
For $4$-quasiplanar graphs the current best bound is $76(n-2)$~\cite{DBLP:journals/dcg/Ackerman09}. 
For the general case, the bounds have been gradually improved from $O(n (\log n)^{O(\log k)})$~\cite{DBLP:journals/combinatorica/PachT97}, to $O(n \log n \cdot 2^{\alpha (n)^c})$~\cite{DBLP:journals/corr/SukW13}.
In case of simple topological drawings, where each pair of edges intersects at most once, a bound of $6.5n + O(1)$ has been proven for $3$-quasiplanar graphs~\cite{DBLP:journals/jct/AckermanT07} and recently $O(n \log n)$ for $k$-quasiplanar graphs with any fixed $k\geq 2$~\cite{DBLP:journals/corr/SukW13}.
It is still open, if the conjecture holds for general $k$.

A $k$-planar graph admits a topological drawing in which each edge has at most $k$ crossings. 
The special case of $1$-planar graphs have been introduced by Ringel~\cite{Ringel65}, who considered the chromatic number of these graphs. 
Important work about the characterization on $1$-planar graphs has been performed by Suzuki~\cite{DBLP:journals/siamdm/Suzuki10}, Thomassen~\cite{DBLP:journals/jgt/Thomassen88a} and Hong \textit{et al.}~\cite{HongELP12}.
Testing $1$-planarity have been shown to be NP-complete for the general case~\cite{Grigoriev07} while efficient algorithms have been found for testing 1-planarity for a given rotation system~\cite{DBLP:journals/tcs/EadesHKLSS13}.
Aspects like straight-line embeddings~\cite{Alametal13} and maximality~\cite{BrandenburgEGGHR12} etc.\ have also been explored.

Closely related to $1$-planar graphs are RAC-drawable graphs~\cite{DBLP:journals/dam/EadesL13,DBLP:journals/jgaa/ArgyriouBKS13}, that is graphs that can be drawn in the plane with straight-line edges and right-angle crossings. For the maximum number of edges in such a graph with $n$ vertices, a bound of $4n -10$ could be proven~\cite{DBLP:journals/tcs/DidimoEL11}, which is remarkably close to the $4n-8$ bound for the class of $1$-planar graphs. A necessary condition for RAC-drawable graph is the absence of fan-crossings. An edge has a $k$-fan-crossing if it crosses $k$ edges that have a common endpoint, cf. Figure~\ref{fig:forbidden-patterns}. RAC-drawings do not allow $2$-fan-crossings. In a recent paper~\cite{Cheongetal13}, Cheong \textit{et al.} considered $k$-fan-crossing free graphs and gave bounds for their maximum number of edges. They obtain a tight bound of $4n-8$ for $n$-vertex $2$-fan-crossing free graphs, and a tight $4n-9$ when edges are required to be straight-line segments.
For $k > 2$, they prove an upper bound of $3(k-1)(n-2)$ edges, while all known examples of $k$-fan-crossing free graphs on $n$ vertices have no more than $kn$ edges.

\paragraph{Our results and more related work.}
As stated before we consider only simple topological drawings, i.e., any two edges have at most one point in common, and only simple graphs, i.e., graphs without loops and parallel edges.
We consider here another variant of sparse non-planar graphs, somehow halfway between $1$-planar graphs and quasiplanar graphs, where we allow more than one crossing on an edge $e$, but only if no endpoint of $e$ is enclosed by $e$ and its crossing edges.
We call this a {\bf fan-crossing} and the class of topological graphs obtained this way {\bf fan-planar graphs}.
Note that we do not differentiate on $k$-fan-crossings as it has been done by Cheong \textit{et al.}~\cite{Cheongetal13}.

The requirement that every edge in $G$ is crossed only by a fan-crossing can be stated in terms of forbidden configurations.
We define \emph{configuration~I} to be one edge that is crossed by two independent edges, and \emph{configuration~II}, respectively \emph{configuration~III}, to be an edge $e$ that is crossed by adjacent edges, such that the triangle (highlighted in figure) formed by the three segments between the crossing points and the common endpoint encloses one of the endpoints of $e$, respectively both endpoints of $e$, see Figure~\ref{fig:crossing-patterns}.
Note that for simple topological drawings, configurations~II and III are well-defined.
Now a simple topological graph is fan-planar if and only if neither configuration~I nor II nor III occurs.
Note that if we forbid only configurations~I and III, then an edge may be crossed by the three edges of a triangle, which is actually not a star, nor a fan-crossing.
However, if every edge is drawn as a straight-line segment, then neither configuration~II nor III can occur and hence in this case it is enough to forbid configuration~I.

\begin{figure}[htb]
 \centering
 \includegraphics{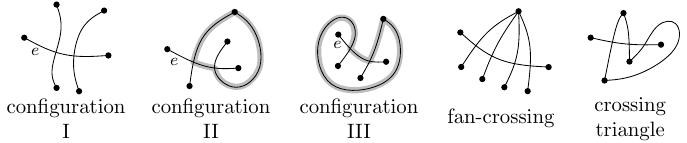}
 \caption{Crossing configurations}
 \label{fig:crossing-patterns}
\end{figure}

Obviously, $1$-planar graphs are also fan-planar.
Furthermore, fan-planar graphs are $3$-quasiplanar since there are no three independent edges that mutually cross.
So, the maximum number of edges in an $n$-vertex fan-planar graph is approximately between $4n$ and $6.5 n$.
In the following, we will explore the exact bound.

\begin{theorem}\label{thm:main}
 Every simple fan-planar graph on $n \geq 3$ vertices has at most $5n-10$ edges. This bound is tight for $n \geq 20$.
\end{theorem}

We remark that fan-planar drawings graphs may have $\Omega(n^2)$ crossings, e.g., a straight-line drawing of $K_{2,n}$ with the bipartition classes placed on two parallel lines.

Very closely related to our approach is the research on forbidden grids in topological graphs, where a $(k,l)$ grid denotes a $k$-subset of the edges pairwise intersected by an $l$-subset of the edges, see~\cite{Pachetal} and~\cite{TardosToth}. It is known that topological graphs without $(k,l)$ grids have a linear number of edges if $k$ and $l$ are fixed. Note that configuration~I, but also a $2$-fan-crossing, is a $(2,1)$ grid. Subsequently~\cite{DBLP:journals/comgeo/AckermanFPS14}, ``natural'' $(k,l)$ grids have been considered, which have the additional requirement that the $k$ edges, as well as the $l$ edges, forming the grid are pairwise disjoint. For natural grids, the achieved bounds are superlinear. Linear bounds on the number of edges have been found for the special case of forbidden natural $(k,1)$ grids where the leading constant heavily depends on the parameter $k$. In particular, the authors give a bound of $65n$ for the case of forbidden natural $(2,1)$ grids, which correspond to our forbidden configuration~I. 
Additionally, the case of geometric graphs, i.e., graphs with straight-line edges, has been explored. For details and differences we refer to~\cite{DBLP:journals/comgeo/AckermanFPS14}. We remark that many arguments in this field of research are based on the probabilistic method, while we use a direct approach aiming on tight upper and lower bounds.

\paragraph{Remark.}
 This paper initiates the study of fan-planar graphs.
 In fact, the first preliminary version of this paper dates back to 2014~\cite{kaufmann2014density}.
 Since then, fan-planarity has become a very popular subject of study among several researchers~\cite{bekos2014recognition,binucci2014fan,di2015planar,hong2015testing,bekos2016density,ackerman2016size,brandenburg2016path,bkr-socg17,DBLP:journals/jgaa/BinucciCDGKKMT17,BekosCGHK17,BiedlC0MNR20,BinucciGDMPST15,DBLP:journals/tcs/Brandenburg20,AngeliniBKKS18,ChimaniKMV19}.
 The interested reader may also have a look at the recent surveys on fan-planarity~\cite{Bekos020} or general beyond-planar graphs~\cite{DidimoLM19}.
 
 Note that in our preliminary version, we excluded configurations I and II only, believing that this would force every edge to be crossed only by a fan-crossing.
 Just recently, we have been contacted by the authors of~\cite{klemz2021}
 with a counterexample to our Lemma~\ref{lem:two-on-cell} in the original context.
 That example (see Figure~\ref{fig:Lemma-1-counterexample_2} for a simpler such example) however contains what we now call configuration III, and is thus not fan-planar for its intended definition.
 Lemma~\ref{lem:two-on-cell} in fact holds true for every fan-planar graph ---We now mention explicitely where the absence of configuration III is used in its proof.

\section{Examples of Fan-Planar Graphs with Many Edges}\label{sec:examples}

The following examples of fan-planar graphs have $n$ vertices and $5n-10$ edges.
The first one results from a $K_{4,n-5}$, where the $n-5$ vertices form a path, see Figure~\ref{fig:K4n-4}.
An easy calculation shows that this graph has $n-1$ vertices and $4(n-5) + (n-6) = 5(n-1) - 21$ edges.
Indeed, one can add $10$ edges to the graph, keeping fan-planarity, as well as additionally one vertex with $6$ more incident edges and obtain a multigraph on $n$ vertices and $5n-10$ edges.
We remark that this graph has parallel edges; however every pair of parallel edges is non-homeomorphic, that is, it surrounds at least one vertex of $G$.
The second example is the (planar) dodecahedral graph where in each $5$-face, we draw $5$ additional edges as a pentagram, see Figure~\ref{fig:dodecahedral}.
This graph has $n = 20$ vertices and $5n-10 = 90$ edges, and has already served as a tight example for $2$-planar graphs~\cite{DBLP:journals/combinatorica/PachT97}.

\begin{figure}[htb]
 \centering
 \subfigure[\label{fig:K4n-4}]{
  \includegraphics{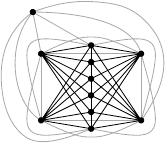}
 }
 \hspace{2em}
 \subfigure[\label{fig:dodecahedral}]{
  \includegraphics{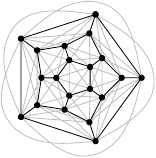}
 }
 \hspace{2em}
 \subfigure[\label{fig:LB}]{
  \includegraphics{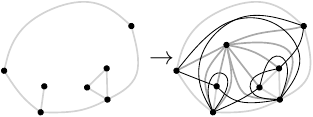}
 }
 \caption{\subref{fig:K4n-4} $K_{4,n-4}$ with $n-4$ vertices on a path. \subref{fig:dodecahedral} The dodecahedral graph with a pentagram in each face. \subref{fig:LB} Adding $2$-hops and spokes into a face.}
\end{figure}

\begin{proposition}\label{prop:LB}
 Every connected planar embedded multigraph $H$ with each face of length at least $3$ can be extended to a fan-planar multigraph $G$ with $5 |V(G)|-10$ edges by adding an independent set of vertices and sufficiently many edges, such that the uncrossed edges of $G$ are precisely the edges of $H$.
 
 Moreover, if $H$ is $3$-connected and each face has length at least $5$, then $G$ has no loops and parallel edges.
\end{proposition}
\begin{proof}
 Let $n$ and $m$ be the number of vertices and edges of $H$, respectively, and $F$ be the set of all faces of $H$. We construct the fan-planar graph $G$ by adding one vertex and two sets of edges into each face $f \in F$. So let $f$ be any face of $H$. Since $H$ is connected, $f$ corresponds to a single closed walk $v_1,\ldots,v_s$, $s \geq 3$, in $H$ around $f$, where vertices and edges may be repeated. 
  We do the following, as illustrated in Figure~\ref{fig:LB}.
 \begin{enumerate}[label=(\arabic*)]
  \item Add a new vertex $v_f$ into $f$.\label{enum:new-vertex}
  \item For $i=1,\ldots,s$ add a new edge $v_fv_i$ drawn in the interior of $f$.\label{enum:new-spokes}
  \item For $i=1,\ldots,s$ add a new edge $v_{i-1}v_{i+1}$ (with indices modulo $s$) crossing the edge $v_fv_i$.\label{enum:new-2-hops}
 \end{enumerate}
 In~\ref{enum:new-vertex} we added $|F|$ new vertices. In~\ref{enum:new-spokes} we added $\deg(f)$ many ``spoke edges'' inside face $f$, in total $\sum_{f}\deg(f) =2m$ new edges. And in~\ref{enum:new-2-hops} we added again $\deg(f)$ many ``$2$-hop edges'' inside face $f$, in total $\sum_{f}\deg(f) =2m$ new edges. Thus we calculate
 $|V(G)| = n + |F|$ and $|E(G)| = m + 2m + 2m = 5m$, 
 which together with Euler's formula $m = n + |F| -2$ gives $|E(G)| = 5|V(G)|-10$. It remains to see that no two edges in $G$ are homeomorphic, and that $G$ is fan-planar. Each ``$2$-hop edge'' $e = v_{i-1}v_{i+1}$ forms a shortcut for a path $v_{i-1}-v_i-v_{i+1}$ on the face $f$.
 Suppose some edge $e'$ in $G$ is parallel to $e$.
 In case $e'$ is also a $2$-hop edge in face $f$, then vertex $v_f$ is surrounded by $e$ and $e'$.
 Otherwise $e'$ is an edge of $H$ or a $2$-hop edge of some other face, and either vertex $v_f$ or vertex $v_i$ is surrounded by $e$ and $e'$.
 
 For the fan-planarity of $G$, observe that each ``spoke edge'' $v_fv_i$ crosses only one $2$-hop edge, and each $2$-hop edge $v_{i-1}v_{i+1}$ crosses only three edges $v_{i-2}v_i$, $v_fv_i$ and $v_iv_{i+2}$, which form a fan-crossing.

 Finally, note that if the planar graph $H$ is $3$-connected and each face has length at least $5$, then the fan-planar graph $G$ has no loops, nor parallel edges, nor crossing incident edges. Examples for such planar graphs are fullerene graphs~\cite{dovslic2002some}.
 Moreover, for every face $f$ in $H$ and the corresponding vertex $v_f$ in $G$ we have $\deg(f) = \deg(v_f)$. So, if every face in $H$ has degree exactly $5$ we can omit all vertices in step~\ref{enum:new-vertex} and edges in~\ref{enum:new-spokes} and obtain a fan-planar graph $G'$ with $V(G') = V(H)$ and $5|V(H)|-10$ edges.
\end{proof}

\section{The $5n-10$ Upper Bound For the Number of Edges}\label{sec:main-proof}

In this section Theorem~\ref{thm:main} is proven. 
It suffices to consider simple topological graphs $G$ that do not contain configuration~I nor~II nor~III and further satisfy the following properties.

\begin{enumerate}[label=(\roman*)]
 \item The chosen embedding of $G$ has the maximum number of uncrossed edges.\label{enum:max-uncrossed}
 \item The addition of any edge to the given embedding violates the fan-planarity of $G$, that is, $G$ is maximal fan-planar with respect to the given embedding.\label{enum:maximal}
\end{enumerate}

\noindent
So for the remainder of this paper let $G$ be a maximal fan-planar graph with a fixed fan-planar embedding with the maximum number of uncrossed edges. Recall that the embedding of $G$ is simple, i.e., any two edges have at most one point in common.

For such a fixed embedding of $G$ we shall split the edges of $G$ into three sets. The first set contains all uncrossed edges. We denote by $H$ the subgraph of $G$ with all vertices in $V$ and all uncrossed edges of $G$. We may refer to $H$ as the \emph{planar subgraph of $G$}. Note that $H$ might be disconnected even if $G$ is connected. In the second set we consider every crossed edge whose endpoints lie in the same connected component of $H$. Note that two such edges may cross each other only if they correspond to the same connected component of $H$. And the third set contains all remaining edges, i.e., every crossed edge with endpoints in different components of $H$. We show how to count the edges in each of the sets and derive the upper bound.

\subsection{Notation, definitions and preliminary results}\label{subsec:preliminaries}

We call a connected component of the plane after the removal of all vertices and edges of $G$ a \emph{cell of $G$}. Whenever we consider a subgraph of $G$ we consider it together with its fan-planar embedding, which is inherited from the embedding of $G$. We will sometimes consider cells of a subgraph $G'$ of $G$, even though those might contain vertices and edges of $G - G'$. The boundary of each cell $c$ is composed of a number of edge segments and some (possibly none) vertices of $G'$. With slight abuse of notation we call the cyclic order of vertices and edge segments along $c$ the \emph{boundary of $c$}, denoted by $\partial c$. Note that vertices and edges may appear more than once in the boundary of a single cell. We define the \emph{size of a cell $c$}, denoted by $||c||$, as the total number of vertices and edge segments in $\partial c$ counted with multiplicity.

Note that by assumptions~\ref{enum:max-uncrossed} and~\ref{enum:maximal} it follows that if two vertices are in the same cell $c$ of $G$ then they are connected by an uncrossed edge of $G$.
However, this uncrossed edge does not necessarily bound cell $c$.

\begin{lemma}\label{lem:two-on-cell}
 If two edges $vw$ and $ux$ cross in a point $p$, no edge at $v$ crosses $ux$ between $p$ and $u$, and no edge at $x$ crosses $vw$ between $p$ and $w$, then $u$ and $w$ are contained in the same cell of $G$.
\end{lemma}
\begin{proof}
 Let $e_0 = ux$ and $e_1 = vw$ be two edges that cross in point $p = p_1$ such that no edge at $v$ crosses $e_0$ between $p_1$ and $u$, and no edge at $x$ crosses $e_1$ between $p_1$ and $w$. If no edge of $G$ crosses $e_0$ nor $e_1$ between $p_1$ and $u$, respectively $w$, then clearly $u$ and $w$ are bounding the same cell. So assume without loss of generality that some edge of $G$ crosses $e_1$ between $p_1$ and $w$. By fan-planarity (i.e., the absence of configuration~I) such edges are incident to $u$ or $x$, where the latter is excluded by assumption. Let $e_2$ be the edge whose crossing with $e_1$ is closest to $w$, and let $p_2$ be the crossing point. As neither configuration~II nor configuration~III occurs, $w$ is not surrounded by $e_0,e_1,e_2$. See Figure~\ref{fig:two-on-cell} for an illustration of the actual situation, and Figures~\ref{fig:Lemma-1-counterexample} and \ref{fig:Lemma-1-counterexample_2} for the potential problematic situations if configuration II or III was allowed.
 \begin{figure}[htb]
  \centering
  \subfigure[\label{fig:two-on-cell}]{
   \includegraphics[scale=0.8]{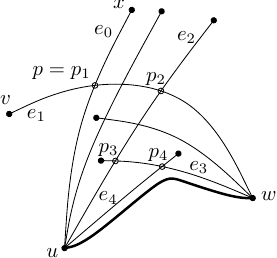}
  }
  \hspace{2em}
 \subfigure[\label{fig:Lemma-1-counterexample}]{
  \includegraphics[scale=0.8]{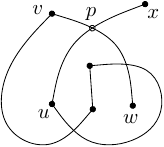}
 }
 \hspace{2em}
 \subfigure[\label{fig:Lemma-1-counterexample_2}]{
  \includegraphics[scale=0.8]{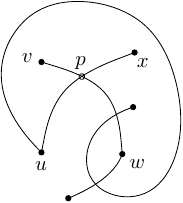}
 }
  \caption{\subref{fig:two-on-cell} Illustration of the proof of Lemma~\ref{lem:two-on-cell}.
  \subref{fig:Lemma-1-counterexample} A counterexample to Lemma~\ref{lem:two-on-cell} when configuration~II is allowed.
  \subref{fig:Lemma-1-counterexample_2} A counterexample to Lemma~\ref{lem:two-on-cell} when configuration~III is allowed.}
 \end{figure}
  
  No edge crosses $e_1$ between $w$ and $p_2$. If $e_2$ is not crossed between $u$ and $p_2$, then $u$ and $w$ are bounding the same cell and we are done. Otherwise let $e_3$ be the edge whose crossing with $e_2$ is closest to $u$, and let $p_3$ be the crossing point. By fan-planarity (the absence of configuration~I) $e_3$ and $e_1$ have a common endpoint, and it is not $v$ since $e_3$ does not cross $e_0$ between $p_1$ and $u$. (Here we use the absence of configuration~II again.) So $e_3$ ends at $w$ and by the absence of configuration~III, $u$ is not surrounded by $e_1,e_2,e_3$.  
  By the choice of $e_3$, we have that $e_2$ is not crossed between $u$ and $p_3$. Again, if $u$ and $w$ are not on the same cell then some edge crosses $e_3$ between $p_3$ and $w$. By fan-planarity (the absence of configuration~I) any such edge $e_4$ has a common endpoint with $e_2$, and if it would not be $u$ then either $e_1$ would be crossed by two independent edges (a configuration~I occurs) or one $v,w$ would be surrounded by $e_2,e_3,e_4$ (a configuration~II or~III occurs) -- a contradiction to the fan-planarity of $G$. So all edges crossing $e_3$ between $w$ and $p_3$ are incident to $u$. Let $e_4$ be such edge whose crossing with $e_3$ is closest to $w$, and let $p_4$ be the crossing point.
  Again, by absence of configuration~III, $v$ and $w$ are not surrounded by $e_2,e_3,e_4$, cf.\ Figure~\ref{fig:two-on-cell} for an illustration.
  
  Iterating this procedure until no edge crosses $e_i$ nor $e_{i-1}$ between $p_i$ and $u,w$ we see that $u$ and $w$ lie indeed on the same cell, which concludes the proof.
 \end{proof}
 
Note that we use the absence of both, configuration~II and configuration~III, in the proof of Lemma~\ref{lem:two-on-cell}.
And indeed, as illustrated in Figure~\ref{fig:Lemma-1-counterexample} and \ref{fig:Lemma-1-counterexample_2}, the statement of the lemma is no longer true if configuration~II or configuration~III may occur.
For better readibility of the remainder, we shall from now on just argue ``by fan-planarity'' without each time explicitely referring to the specific forbidden configurations.
Conclusively, in a fan-planar drawing each edge is either uncrossed or crossed by a fan-crossing.

\begin{figure}[htb]
 \centering
 \subfigure[\label{fig:4-cell}]{
  \includegraphics[scale=0.8]{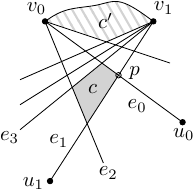}
 }
 \hspace{1.8em}
 \subfigure[\label{fig:4-cell-2}]{
  \includegraphics[scale=0.8]{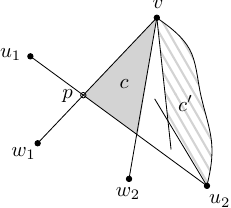}
 }
 \hspace{1.8em}
 \subfigure[\label{fig:crossing-stars-2}]{
  \includegraphics[scale=0.8]{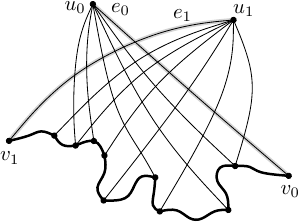}
 }
 \caption{Illustration of the proofs of \subref{fig:4-cell}, \subref{fig:4-cell-2} Corollary~\ref{cor:4cells-are-empty}
 and \subref{fig:crossing-stars-2} Corollary~\ref{cor:crossing-stars}.
 }
\end{figure}

Lemma~\ref{lem:two-on-cell} has a couple of nice consequences.

\begin{corollary}\label{cor:uncrossed-edge}
 Any two crossing edges in $G$ are connected by an uncrossed edge.
\end{corollary}
\begin{proof}
 Let $ux$ and $vw$ be the two crossing edges. By fan-planarity either no other edge at $x$ or no other edge at $u$ crosses the edge $vw$, say there is no such edge at $x$. Similarly, we may assume without loss of generality that no edge at $v$ crosses the edge $ux$. However, this implies that $ux$ and $vw$ satisfy the requirements of Lemma~\ref{lem:two-on-cell} and we have that $u$ and $w$ are on the same cell. In particular, we can draw an uncrossed edge between $u$ and $w$ in this cell. Because $G$ is maximally fan-planar, $uw$ is indeed an edge of $G$. And since $G$ is embedded with the maximum number of uncrossed edges, $uw$ is also drawn uncrossed.
\end{proof}

\begin{corollary}\label{cor:4cells-are-empty}
 If $c$ is a cell of any subgraph of $G$, and $||c||= 4$, then $c$ contains no vertex of $G$ in its interior.
\end{corollary}

\begin{proof}
 Let $c$ be a cell of $G' \subseteq G$ with $||c||=4$. Then $\partial c$ consists either of four edge segments or one vertex and three edge segments. Let us assume for the sake of contradiction that $c$ contains a set $S \neq \emptyset$ of vertices in its interior.
 
 \medskip
 
 \noindent
 \textit{Case 1. $\partial c$ consists of four edge segments.} Let $e_0,e_1,e_2,e_3$ be the edges bounding $c$ is this cyclic order. From the fan-planarity of $G$ follows that $e_0$ and $e_2$ have a common endpoint $v_0$. Similarly $e_1$ and $e_3$ have a common endpoint $v_1$. See Figure~\ref{fig:4-cell} for an illustration. If $p$ denotes the crossing point of $e_0 = v_0u_0$ and $e_1 = v_1u_1$, then by fan-planarity no edge at $u_i$ crosses $e_{i+1}$ between $p$ and $v_{i+1}$, where $i \in \{0,1\}$ and indices are taken modulo~$2$. Hence by Lemma~\ref{lem:two-on-cell} there exists a cell $c'$ of $G$ that contains both $v_0$ and $v_1$.
 
 Now consider the subgraph $G[S]$ of $G$ on the vertices inside $c$. From the fan-planarity follows that every edge between $G[S]$ and $G[V\setminus S]$ has as one endpoint $v_0$ or $v_1$. We now change the embedding of $G$ by placing the subgraph $G[S]$ (keeping its inherited embedding) into the cell $c'$ that contains $v_0$ and $v_1$. The resulting embedding of $G$ is still fan-planar and moreover at least one edge between $G[S]$ and $\{v_0,v_1\}$ is now uncrossed -- a contradiction to our assumption~\ref{enum:max-uncrossed} that the embedding of $G$ has the maximum number of uncrossed edges.
 
 \medskip
 
 \noindent
 \textit{Case 2. $\partial c$ consists of one vertex and three edge segments.} Let $v$ be the vertex and $vw_1$, $vw_2$, $u_1u_2$ be the edges bounding $c$. See Figure~\ref{fig:4-cell-2} for an illustration. If $p$ denotes the crossing point of $vw_1$ and $u_1u_2$, then by fan-planarity either no edge at $u_1$ crosses $vw_1$ between $p$ and $v$ or no edge at $u_2$ crosses $vw_1$ between $p$ and $v$. Moreover, for $i=1,2$ the edge $vw_i$ is the only edge at $w_i$ that crosses $u_1u_2$. Hence by Lemma~\ref{lem:two-on-cell} we have that either $v$ and $u_1$ or $v$ and $u_2$ are contained in the same cell of $G$ -- say cell $c'$ contains $v$ and $u_2$.
 
 Now, similarly to the previous case, consider the subgraph $G[S]$ of $G$ on the vertices inside $c$. From the fan-planarity, it follows that every edge between $G[S]$ and $G[V\setminus S]$ has as one endpoint $v$, $u_1$ or $u_2$. Moreover, every edge between a vertex in $G[S]$ and $u_1$ or $u_2$ is crossed only by edges incident to $v$, as otherwise $u_1u_2$ would be crossed by two independent edges. We now change the embedding of $G$ by placing the subgraph $G[S]$ (keeping its inherited embedding) into the cell $c'$ that contains $v$ and $u_2$. The resulting embedding of $G$ is still fan-planar and moreover at least one edge between $G[S]$ and $u_2$ is now uncrossed -- a contradiction to~\ref{enum:max-uncrossed}.
\end{proof}

\begin{corollary}\label{cor:crossing-stars}
 If $e_0 = u_0v_0$ and $e_1 = u_1v_1$ are two crossing edges of $G$ such that every edge of $G$ crossing $e_i$ is crossed only by edges incident to $u_{i+1}$, where $i \in \{0,1\}$ and indices are taken modulo $2$, then $v_0$ and $v_1$ are in the same connected component of $H$.
\end{corollary}

\begin{proof}
 Let $p$ be the point in which $e_0$ and $e_1$ cross. For $i=0,1$ let $S_i$ be the set of all edges crossing $e_{i+1}$ between $p$ and $v_{i+1}$. (All indices are taken modulo $2$.) By assumption $S_i$ is a star
 centered at $u_i$. Consider the embedding of the graph $S_0 \cup S_1$ inherited from $G$. By fan-planarity (i.e., the absence of both, configuration~I and~II) $u_0$ and $u_1$ are contained in the outer cell of $S_0 \cup S_1$. Moreover, every inner cell $c$ of $S_0 \cup S_1$ has $||c||=4$ and thus by Corollary~\ref{cor:4cells-are-empty} all leaves of $S_0$ and $S_1$ are also contained in the outer cell $c^*$ of $S_0 \cup S_1$.
 
 We claim that no edge segment in the boundary $\partial c^*$ of the outer cell is crossed by another edge in $G$. Indeed, if $e'$ is an edge crossing some edge $e \in S_0 \cup S_1$ between the crossing of $e$ and $e_0$ or $e_1$ and the endpoint of $e$ different from $u_0,u_1$, then by assumption one endpoint of $e'$ is $u_0$ or $u_1$ -- say $u_1$. Moreover, since by Corollary~\ref{cor:4cells-are-empty} no cell $c$ with $||c||=4$ contains any vertex, we have that $e'$ crosses $e_0$ between $p$ and $v_0$ and thus $e \in S_1$. See Figure~\ref{fig:crossing-stars-detail-2}.
 
 \begin{figure}[htb]
  \centering
  \subfigure[\label{fig:crossing-stars}]{
   \includegraphics{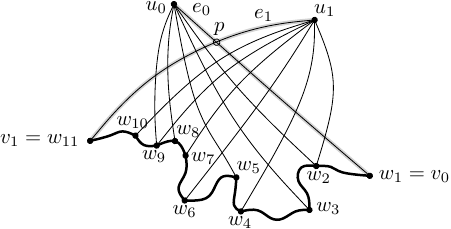}
  }
  \hspace{2em}
  \subfigure[\label{fig:crossing-stars-detail-2}]{
   \includegraphics{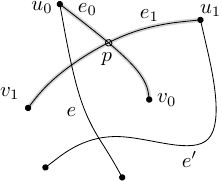}
  }
  \caption{\subref{fig:crossing-stars} The stars $S_0$ and $S_1$ in the proof of Corollary~\ref{cor:crossing-stars}. \subref{fig:crossing-stars-detail-2} If an edge $e'$ crosses $e \in S_0$ between the crossing of $e$ and $e_1$ and the endpoint of $e$ different from $u_0$, and $e' \notin S_1$, then $v_0$ is contained in a cell $c$ bounded by $e,e'$ and $e_1$ with $||c||=4$.}
 \end{figure}
 
 We conclude that if we label the vertices of $S_0 \cup S_1$ such that their cyclic order around $c^*$ is $u_0,u_1,v_0=w_1,w_2,\ldots,w_k=v_1$, then for each $j \in \{1,\ldots,k-1\}$ the vertices $w_j$ and $w_{j+1}$ are contained in the same cell of $G$ and hence by maximality of $G$ joint by an uncrossed edge. See Figure~\ref{fig:crossing-stars} for an illustration.
\end{proof}

Recall that $H$ denotes the planar subgraph of $G$. For convenience we refer to the closure of cells of $H$ as the \emph{faces of $G$}. The boundary of a face $f$ is a disjoint set of (not necessarily simple) closed walks in $H$, which we call \emph{facial walks}. The \emph{length of a facial walk $W$}, denoted by $|W|$, is the number of its edges counted with multiplicity. We remark that a facial walk may consist of only a single vertex, in which case its length is $0$. See Figure~\ref{fig:H-boundary} for an example.

For a face $f$ and a facial walk $W$ of $f$, we define $G(W)$ to be the subgraph of $G$ consisting of the walk $W$ and all edges that are drawn entirely inside $f$ and have both endpoints on $W$. The set of cells of $G(W)$ that lie inside $f$ is denoted by $C(W)$. Finally, the graph $G(W)$ is called a \emph{sunflower} if $|W| \geq 5$ and $G(W)$ has exactly $|W|$ inner edges each of which connects two vertices at distance~$2$ on $W$. See Figure~\ref{fig:maximal_faces} for an example of a sunflower. We remark that for convenience we depict facial walks in our figures as simple cycles, even when there are repeated vertices or edges. Indeed, we can assume facial walks to be simple cycles as long as we bound the number of edges in terms of the length of facial walks and sizes of cells. Only in the final proof of Theorem~\ref{thm:main} we bound the number of edges in terms of the number of vertices, and there those repetitions will be taken into account by Euler's formula.

\begin{figure}[htb]
 \centering
 \subfigure[\label{fig:H-boundary}]{
  \includegraphics{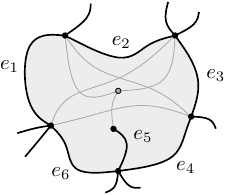}
 }
 \hspace{3em}
 \subfigure[\label{fig:maximal_faces}]{
  \includegraphics{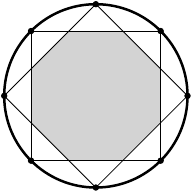}
 }
 \caption{\subref{fig:H-boundary} A cell of $H$ (drawn black) is shown in gray. The boundary of the cell is the union of the closed walk $e_1,e_2,e_3,e_4,e_5,e_5,e_6$ and the single vertex in the interior. \subref{fig:maximal_faces} A sunflower on $8$ vertices. The facial walk $W$ is drawn thick. A cell bounded by $8$ edge segments and no vertex is highlighted.}
\end{figure}

\subsection{Counting the Number of Edges}\label{subsec:counting-edges}

We shall count the number of edges of $G$ in three sets:
\begin{itemize}
 \item Edges in $H$, that is all uncrossed edges.
 \item Edges in $E(G(W)) \setminus E(W)$ for every facial walk $W$.
 \item Edges between different facial walks of the same face $f$ of $G$.
\end{itemize}

The edges in $H$ will be counted in the final proof of Theorem~\ref{thm:main} below. We start by counting the crossed edges, first within the same facial walk and afterwards between different facial walks. For convenience, for a facial walk $W$ the edges in $E(G(W)) \setminus E(W)$ and their edge segments are called \emph{inner edges} and \emph{inner edge segments} of $G(W)$, respectively.

\begin{lemma}\label{lem:condition-sunflower}
 Let $W$ be any facial walk. If every inner edge segment of $G(W)$ bounds a cell of $G(W)$ of size $4$ and no cell of $G(W)$ contains two vertices on its boundary not consecutive in $W$, then $G(W)$ is a sunflower.
\end{lemma}

\begin{proof}
 Let $v_0,\ldots,v_k$ be the clockwise order of vertices around $W$. (In the following, indices are considered modulo $k+1$.) For any vertex $v_i$ we consider the set of inner edges incident to $v_i$. Since no two non-consecutive vertices of $W$ lie on the same cell, every $v_i$ has at least one such edge. Moreover, note that for each edge $v_iv_{i+1}$ of $W$ the unique cell $c_i$ with $v_iv_{i+1}$ on its boundary has size at least $5$. This implies that every $v_i$ has indeed at least two incident inner edges. Finally, note that every inner edge is crossed, since otherwise there would be two non-consecutive vertices of $W$ bounding the same cell of $G(W)$.

 Now let us consider the clockwise first inner edge incident to $v_i$, denoted by $e^1_i$. Since an edge segment of $e^1_i$ bounds the cell $c_i$ on the other side of this segment is a cell of size $4$. This means that $e^1_i$ and the clockwise next inner edge at $v_i$ are crossed by some edge $e$. By fan-planarity $e$ crosses only edges incident to $v_i$. Thus each endpoint of $e$ bounds together with $v_i$ some cell of $G(W)$. Since only consecutive vertices of $W$ bound the same cell of $G(W)$, this implies that $e = v_{i-1}v_{i+1}$. Since this is true for every $i \in \{0,\ldots,k\}$, we conclude that $G(W)$ is a sunflower.
\end{proof}

Recall that $C(W)$ denotes the set of all bounded cells of $G(W)$.

\begin{lemma}\label{lem:inner-face}
 For every facial walk $W$ with $|W| \geq 3$ we have
 \[|E(G(W)) \setminus E(W)| \leq 2|W| - 5 - \sum_{c \in C(W)}\max\{0,||c||-5\}.\]
\end{lemma}
\begin{proof}
  Without loss of generality we may assume that $W$ is a simple cycle, since we bound the number of \emph{inner} edges of $W$ in terms of the length of $W$.  
 We proceed by induction on the number of inner edges. As induction base we consider the case that $E(G(W)) \setminus E(W) = \emptyset$. Then $G(W) = W$ and $C(W)$ consists of a single cell $c$ with $||c||=2|W|$. Thus
 \[
 |E(G(W)) \setminus E(W)| = 0 = 2|W| - 5 - (||c||-5).
 \]
 
 So assume that there is at least one inner edge.
 First, consider any inner edge segment $e^*$ and the two cells $c_1,c_2 \in C(W)$ containing $e^*$ on their boundary. If $c^*$ denotes the cell $c_1 \cup c_2$ of $G(W) \setminus e$, then
 \[
  ||c^*|| = ||c_1||+||c_2||-4
 \]
 and thus
 \begin{equation}
  \max\{0,||c^*||-5\} = \max\{0,||c_1||-5\} + \max\{0,||c_2||-5\} + x,\label{eq:split-cells}
 \end{equation}
 where $x = 1$ if $||c_1|| \geq 5$ and $||c_2|| \geq 5$ and $x=0$ otherwise.

 Now, we shall distinguish three cases: $G(W)$ is a sunflower, some inner edge segment is not bounded by a cell of size $4$, and some cell of $G(W)$ contains two vertices on its boundary that are not consecutive in $W$. By Lemma~\ref{lem:condition-sunflower} this is a complete case distinction.

 \begin{enumerate}[itemsep = \medskipamount, itemindent = 2.2em, label = \textit{Case \arabic*.}]
  \item \textit{$G(W)$ is a sunflower.} Then by definition, $G(W)$ has exactly $|W|$ inner edges. Moreover, $C(W)$ contains exactly one cell $c$ of size greater than $4$ and for that cell we have $||c|| = |W|$. Thus
  \[
   |E(G(W)) \setminus E(W)| = |W| = 2|W| - 5 - (|W|-5).
  \]
 
  \item \textit{Some edge segment $e^*$ of some inner edge $e$ bounds two cells $c_1, c_2$ of size at least $5$ each.} Then applying induction to the graph $G' = G(W) \setminus e$ we get
  \begin{align*}
   |E(G(W)) \setminus E(W)| &= 1 + |E(G') \setminus E(W)| \\
   &\leq 1 + 2|W| - 5 - \sum_{c \in C(G')} \max\{0,||c||-5\}\\
   &\overset{\eqref{eq:split-cells}}{\leq} 1 + 2|W| - 5 - \sum_{c \in C(W)} \max\{0,||c||-5\} - 1.
  \end{align*}
  
  \item \textit{Some cell of $G(W)$ contains two vertices $u,w$ on its boundary that are not consecutive on $W$.}
  Note that $uw$ may or may not be an inner edge of $G(W)$. In the latter case we denote by $c^*$ the unique cell that is bounded by $u$ and $w$. In any case exactly two cells $c_1,c_2$ of $G(W) \cup uw$ are bounded by $u$ and $w$ and we have $||c^*|| = ||c_1|| + ||c_2|| - 4$, provided $c^*$ exists.

  \begin{figure}[htb]
   \centering
   \includegraphics[scale=0.8]{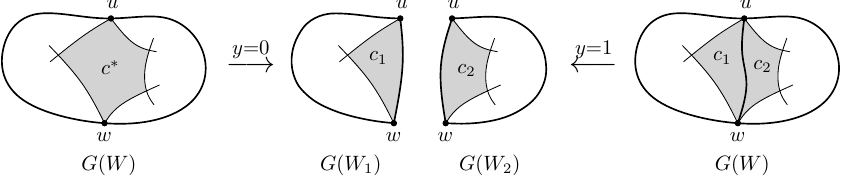}
   \caption{The graph $G(W)$ is split into two graphs $G(W_1)$ and $G(W_2)$ along two vertices $u,w$ that are not consecutive on $W$ but bound the same cell of $G(W)$.}
   \label{fig:inner-face-split}
  \end{figure}
 
  We consider the two cycles $W_1,W_2$ in $W \cup uw$ that are different from $W$, such that $W_1$ surrounds $c_1$ and $W_2$ surrounds $c_2$. For $i=1,2$ consider $G(W_i)$, i.e., the subgraph of $G(W) \cup uw$ induced by $W_i$, see Figure~\ref{fig:inner-face-split}. We have
  \begin{align*}
   |W| & = |W_1| + |W_2| - 2,\\
   |E(G(W)) \setminus E(W)| &= |(E(G(W_1)) \setminus E(W_1))|\\
   & \qquad + |(E(G(W_2)) \setminus E(W_2))| + y,\\
   \sum_{c \in C(W)}\max\{0,||c||-5\} &\overset{\eqref{eq:split-cells}}{=} \sum_{c \in C(W_1)}\max\{0,||c||-5\}\\
   & \qquad + \sum_{c \in C(W_2)}\max\{0,||c||-5\} + (1-y),
  \end{align*}
  where $y = 1$ if $uw$ already was an inner edge of $G(W)$ and $y=0$ otherwise. Now, applying induction to $G(W_1)$ and $G(W_2)$ gives the claimed bound.
 \end{enumerate}
 \vspace{-1em}
\end{proof}

Let us define by $C(f)$ the union of $C(W)$ for all facial walks $W$ of $f$. Moreover, we partition $C(f)$ into $C_\emptyset(f)$ and $C_*(f)$, where a cell $c \in C(f)$ lies in $C_\emptyset(f)$ if and only if $(c \setminus \partial c) \cap V(G) = \emptyset$. I.e., cells in $C_\emptyset(f)$ do not have any vertex of $G$ in their open interior, whereas cells in $C_*(f)$ contain some vertex of $G$ in their interior. Without loss of generality we have that for each $f$, $C_*(f)$ is either empty or contains at least one bounded cell. This can be achieved by picking a cell of $G$ that has the maximum number of surrounding Jordan curves of the form $\partial c$ for $c \in \bigcup_f C_*(f)$, and defining it to be in the unbounded cell of $G$.

Before we bound the number of edges between different facial walks of $f$ we need one more lemma. Consider a face $f$ of $G$ with at least two facial walks and a cell $c \in C_*(f)$ that is inclusion-minimal. Let $W_1$ be the facial walk with $c \in C(W_1)$ and $W_2,\ldots,W_k$ be the facial walks that are contained in $c$. For $i=1,\ldots,k$ let $c_i$ be the cell of $G(W_i)$ that contains all walks $W_j$ with $j \neq i$. In particular, we have $c_1 = c$. Moreover, we call an edge between two distinct facial walks $W_i$ and $W_j$ a \emph{$W_iW_j$-edge}.

\begin{lemma}\label{lem:no-two-vertices}
 Exactly one of $c_1,\ldots,c_k$ has a vertex on its boundary.
\end{lemma}
\begin{proof}
 We proceed by proving a series of claims first.
 
 \begin{ourclaim}\label{claim:between-same-walks}
  If a $W_iW_j$-edge and a $W_{i'}W_{j'}$-edge cross, then $\{i,j\} = \{i',j'\}$.
 \end{ourclaim}
 \begin{claimproof}
  Consider a $W_iW_j$-edge $e_1 = u_1v_1$ crossing a $W_{i'}W_{j'}$-edge $e_2 = u_1v_2$. By Corollary~\ref{cor:uncrossed-edge} one endpoint of $e_1$, say $u_1 \in W_i$, and one endpoint of $e_2$, say $u_2 \in W_{i'}$, are joint by an uncrossed edge. In particular, $W_i = W_{i'}$.

  If, Case 1, $e_1$ is crossed by a second edge incident to $v_2$, then applying Lemma~\ref{lem:two-on-cell} gives an uncrossed edge $u_1v_2$, which is a contradiction to the fact that $W_{j'} \neq W_{i'}$, or an uncrossed edge $v_1v_2$, which implies $W_j = W_{j'}$ as desired. See Figure~\ref{fig:between-same-walks}.

  Otherwise, Case 2, $e_1$ is crossed only by edges at $u_2$, and by symmetry $e_2$ is crossed only by edges at $u_1$.
  Applying Corollary~\ref{cor:crossing-stars} we get that $v_1$ and $v_2$ are in the same connected component of $H$ and hence $W_j = W_{j'}$, as desired.
 \end{claimproof}

 \begin{figure}[htb]
  \centering
  \subfigure[\label{fig:between-same-walks}]{
   \includegraphics[scale=0.8]{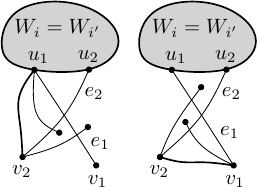}
  }
  \hspace{1em}
  \subfigure[\label{fig:open-one-side}]{
   \includegraphics[scale=0.8]{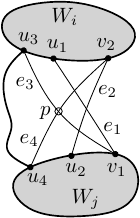}
  }
  \hspace{1em}
  \subfigure[\label{fig:to-next-end}]{
   \includegraphics[scale=0.8]{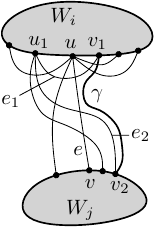}
  }
  \hspace{1em}
  \subfigure[\label{fig:one-open-to-closed}]{
   \includegraphics[scale=0.8]{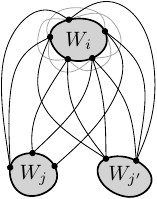}
  }
  \caption{\subref{fig:between-same-walks} Case 1 in the proof of Claim~\ref{claim:between-same-walks}. Illustrations of the proofs of \subref{fig:open-one-side} Claim~\ref{claim:open-one-side}, \subref{fig:to-next-end} Claim~\ref{claim:to-next-end} and \subref{fig:one-open-to-closed} Claim~\ref{claim:one-open-to-closed}.}
 \end{figure}

 For a facial walk $W_i$ a vertex $v \in W_i$ is called \emph{open} if $v$ lies on $\partial c_i$. Moreover, a vertex $v \in W_i$ is called \emph{closed} if $v$ is not open but there is at least one edge between $v$ and another facial walk $W_j \neq W_i$. So the endpoints of every $W_iW_j$-edge are open or closed, and by fan-planarity (the absence of configuration~I) at least one endpoint is open.

 \begin{ourclaim}\label{claim:open-one-side}
  If two $W_iW_j$-edges cross then both have exactly one open endpoint, which moreover are in the same facial walk.
 \end{ourclaim}
 \begin{claimproof}
  Let $e_1 = u_1v_1$ and $e_2 = u_2v_2$ be two crossing $W_iW_j$-edges. Assume for the sake of contradiction that $e_1$ has an open endpoint $u_1 \in W_i$ and $e_2$ has an open endpoint $u_2 \in W_j$. We consider the edges $e_3 = u_3v_1$ and $e_4 = u_4v_2$ that are incident to $v_1$ and $v_2$ respectively, cross each other and whose crossing point $p$ is furthest away from $v_1$ and $v_2$. See Figure~\ref{fig:open-one-side} for an illustration. Note that possibly $e_1 = e_3$ and/or $e_4 = e_2$.

  Now $u_3$ is not in $W_j$ because $u_1$ is an open endpoint and $u_4$ is not in $W_i$ because $u_2$ is an open end. Hence by Claim~\ref{claim:between-same-walks} $u_3 \in W_i$ and $u_4 \in W_j$. Moreover, by Lemma~\ref{lem:two-on-cell} $u_3u_4$ is an uncrossed edge of $G$ -- a contradiction to the fact that $W_i$ and $W_j$ are distinct facial walks.
 \end{claimproof}

 Claim~\ref{claim:open-one-side} implies that every edge between different facial walks has exactly one open endpoint and one closed endpoint.
 Indeed, by assumption~\ref{enum:max-uncrossed} no cell of $G$ has vertices from distinct facial walks on its boundary, meaning that every $W_iW_j$-edge $e$ is crossed by some other edge $e'$.
 If $e'$ runs between two facial walks, then $e$ has a closed endpoint by Claim~\ref{claim:open-one-side}, and if $e'$ runs within $W_i$ or $W_j$, then $e$ has a closed endpoint by definition.
 Let us remark that the second part of Claim~\ref{claim:open-one-side} uses Lemma~\ref{lem:two-on-cell}, which in turn uses the absence of configuration~II, and indeed, the second part of Claim~\ref{claim:open-one-side} is no longer true if configuration~II may occur.

 \begin{ourclaim}\label{claim:to-next-end}
  If a $W_iW_j$-edge has a closed endpoint $u \in W_i$ and $w$ is the counterclockwise next open or closed vertex of $W_i$ after $u$, then there exists a $W_iW_j$-edge incident to $w$ with open endpoint in $W_j$.
 \end{ourclaim}
 \begin{claimproof}
  Let $e = uv$ be a $W_iW_j$-edge that has a closed endpoint $u \in W_i$.
  By fan-planarity (absence of configuration~I) $v$ is an open vertex of $W_j$.
  As $u$ is a closed endpoint, some edge of $G(W_i)$ crosses $e$.
  Let $e_1 = u_1v_1$ be the edge from $G(W_i)$ whose crossing with $e$ is closest to $v$ (meaning that this crossing lies on $\partial c_i$), where without loss of generality $v_1$ comes counterclockwise after $u$ in $W_i$.
  Further assume without loss of generality that $e$ is the $W_iW_j$-edge at $u$ whose crossing with $e_1$ is closest to $v_1$.
  
  Consider a Jordan curve $\gamma$ between $v_1$ and the crossing of $e_1$ and $e$ that runs along the left side of $e_1$ and is crossed only by those edges crossing $e_1$ on this stretch.
  Below we show how to extend $\gamma$ without creating new crossings so that $\gamma$ ends at some open vertex in $W_j$ and argue that $G$ together with $\gamma$ seen as an edge at $v_1$ is fan-planar.
  Then the edge-maximality of $G$ implies that an edge with the same endpoints as $\gamma$ already exists in $G$, which will prove the claim.
  Note that by fan-planarity (absence of configuration~I) every edge crossing $\gamma$ is incident to $v$ or $u$.
  If one such edge is incident to $u$, it is the desired $W_iW_j$-edge.
  Otherwise, all such edges are incident to $u$ and by the choice of $e$ the other endpoint lies also in $W_i$.

  So let us first assume that $e$ is not crossed between $v$ and its crossing with $e_1$.
  In this case we can easily extend $\gamma$ to end at $v$ and we are done.

  So $e$ is crossed between its crossing with $e_1$ and $v$.
  Let $e_2$ be such a crossing edge whose crossing with $e$ is closest to $u$.
  Then by fan-planarity (absence of configuration~I) $e_2$ is incident to $u_1$ or $v_1$.
  Moreover, by Claim~\ref{claim:between-same-walks} and Claim~\ref{claim:open-one-side} $e_2$ has a closed endpoint in $W_i$ and an open endpoint in $W_j$.
  Thus if $e_2$ is incident to $v_1$, then $e_2$ is the desired $W_iW_j$-edge.
  So assume that $e_2 = u_1v_2$ for some $v_2 \in W_j$.
  We extend $\gamma$ along the left side of $e$ and $e_2$ all the way to $v_2$.
  We refer to Figure~\ref{fig:to-next-end} for an illustration.
  If $\gamma$ is not crossed on this stretch, we are done.
  Otherwise, by the choice of $e_2$, $\gamma$ is crossed while running along $e_2$.
  Let $e_3$ be such a crossing edge.
  By fan-planarity, $e_3$ ends at $u$ or $v$ and by the choice of $e$, it does not end at $u$.
  Finally, by Claim~\ref{claim:between-same-walks} the endpoint of $e_3$ different from $v$ lies in $W_i$, which makes $e_3$ our desired edge.
 \end{claimproof}
 
 Claim~\ref{claim:to-next-end} together with Claim~\ref{claim:open-one-side} implies that on each facial walk every closed vertex is followed by another closed vertex. In particular, the facial walks come in two kinds, one with open vertices only and one with closed vertices only. We remark that one can show that, if $W_i$ has only closed vertices, then $G(W_i)$ is a sunflower.
 
 \begin{ourclaim}\label{claim:one-open-to-closed}
  Every facial walk with only closed vertices has edges to exactly one facial walk with only open vertices.
 \end{ourclaim}
 \begin{claimproof}
  Assume for the sake of contradiction that facial walk $W_i$ with only closed vertices has edges to two different facials walks $W_j,W_{j'}$ with only open vertices. Claim~\ref{claim:to-next-end} implies that if some closed vertex of $W_i$ has an edge to $W_j$, then every closed vertex of $W_i$ has an edge to $W_j$, and the same is true for $W_{j'}$. Hence, each of the at least three closed vertices in $W_i$ has an edge to $W_j$ and an edge to $W_{j'}$, which implies that some $W_iW_j$-edge and some $W_iW_{j'}$-edge must cross, see Figure~\ref{fig:one-open-to-closed}. (Indeed, if any two such edges would not cross, then contracting $W_j$ and $W_{j'}$ into a single point each and placing a new vertex in the middle of $W_i$ with an edge to every closed vertex in $W_i$ would give a planar drawing of $K_{3,3}$.) Thus by Claim~\ref{claim:between-same-walks} we have $W_j = W_{j'}$ -- a contradiction to our assumption.
 \end{claimproof}

 \medskip
 
 \noindent
 We are now ready to prove that at most one facial walk has open vertices. Recall that by Claim~\ref{claim:to-next-end} every facial walk is of one of two kinds: only open vertices or only closed vertices. Moreover, by fan-planarity (absence of configuration~I) and Claim~\ref{claim:open-one-side} no edge runs between two facial walks of the same kind. We consider a bipartite graph $F$ whose black and white vertices correspond to facial walks of the first and second kind, respectively, and whose edges correspond to pairs $W_i,W_j$ of facial walks for which there is at least one $W_iW_j$-edge in $G$. Since $G$ is connected, $F$ is connected, and by Claim~\ref{claim:one-open-to-closed} every white vertex is adjacent to exactly one black vertex. This means that $F$ is a star and has exactly one black vertex, which concludes the proof.
\end{proof}

Now, we can bound the number of $W_iW_j$-edges.
Recall that $c$ is an in\-clu\-sion-minimal cell in $C_*(f)$ for some face $f$ of $G$, $W_1$ denotes the facial walk of $f$ with $c \in C(W_1)$ and $W_2,\ldots,W_k$ denote the facial walks of $f$ inside $c$.
Further, for $i=1,\ldots,k$ we denote by $c_i$ the cell of $G(W_i)$ containing all $W_j$ with $j\neq i$.

\begin{lemma}\label{lem:edges-between}
 The number of edges between $W_1,\ldots,W_k$ is at most
 \[
  4(k-2) + \sum_{i=1}^k||c_i||.
 \] 
\end{lemma}
\begin{proof}
 By Lemma~\ref{lem:no-two-vertices} exactly one of $c_1,\ldots,c_k$ has vertices on its boundary, say $W_1$. Let $U$ be the set of vertices on the boundary of $c_1$. For a vertex $u \in U$ and an index $i\in\{2,\ldots,k\}$ we call an edge between $u$ and $W_i$ a \emph{$uW_i$-edge}. We define a bipartite graph $J$ as follows. One bipartition class is formed by the vertices in $U$. In the second bipartition class there is one vertex $w_i$ for each facial walk $W_i$, $i=1,\ldots,k$. A vertex $u \in U$ is connected by an edge to $w_i$ if and only if $i=1$ or $i \geq 2$ and there is a $uW_i$-edge.
 
 \begin{ourclaim}\label{claim:planar-bipartite}
  The graph $J$ is planar.
 \end{ourclaim}
 \begin{claimproof}
  We consider the following embedding of $J$. Afterwards we shall argue that this embedding is indeed a plane embedding. So take the position of every vertex $u \in U$ from the fan-planar embedding of $G$. For $i \geq 2$, we consider the drawing of $W_i$ in the embedding of $G$, for each edge between a vertex $u \in U$ and the vertex $w_i$ in $J$ we take the drawing of one $uW_i$-edge in $G$, and then contract the drawing of $W_i$ into a single point -- the position for vertex $w_i$. Finally, we place the last vertex $w_1$ outside the cell $c_1$ and connect $w_1$ to each $u \in U$ in such a way that these edges do not cross any other edge in $J$. See Figure~\ref{fig:drawing-of-I} for an illustrating example.
  
  \begin{figure}[htb]
   \centering
   \subfigure[\label{fig:drawing-of-I}]{
    \includegraphics[scale=0.8]{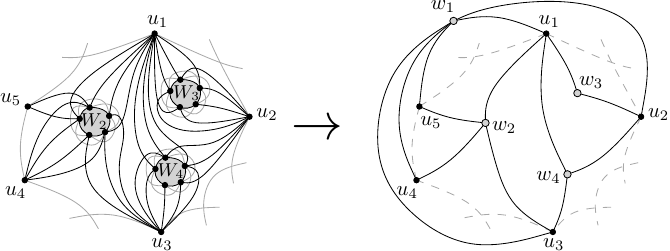}
   }
   \hspace{2em}
   \subfigure[\label{fig:deg-claim}]{
    \includegraphics[scale=0.8]{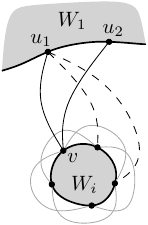}
   }
   \caption{\subref{fig:drawing-of-I} Obtaining the graph $J$. \subref{fig:deg-claim} The contradiction in Claim~\ref{claim:uW_i-edges}.}
   
  \end{figure}

  Now the resulting drawing of $J$ contains crossing edges only if a $uW_i$-edge crosses a $u'W_{i'}$-edge in $G$. However, by Lemma~\ref{lem:no-two-vertices} the cells $c_2,\ldots,c_k$ have no vertices on their boundary. Hence, for each $i=2,\ldots,k$ every $uW_i$-edge crosses an edge of $G(W_i)$. Now if a $uW_i$-edge $e$ would cross a $u'W_{i'}$-edge with $u \neq u'$ and $i \neq i'$, then $e$ would be crossed by two independent edges, which is configuration~I and hence a contradiction to the fan-planarity of $G$.
 \end{claimproof}
 
 \noindent
 Since $J$ is a planar bipartite graph with bipartition classes of size $|U|$ and $k$ we have
 \[
  |E(J)| = \sum_{i=1}^k \deg_J(w_i) \leq 2(|U|+k)-4.
 \]
 
 \begin{ourclaim}\label{claim:uW_i-edges}
  For each $i = 2,\ldots,k$ the number of $uW_i$-edges is at most
  \[
   ||c_i|| + 2\deg_J(w_i).
  \]
 \end{ourclaim}
 \begin{claimproof}
  Consider the vertices on $W_i$ and the set $U' \subseteq U$ of vertices on $W_1$ that have a neighbor on $W_i$. 
  For each $u \in U'$ consider the cyclic ordering of $uW_i$-edges around $u$.
  Since not every edge at $u$ is a $uW_i$-edge (at least one edge ends in $W_1$) we obtain a linear order on the $uW_i$-edges going counterclockwise around $u$.
  
  Now we claim that when we remove for each $u \in U'$ the last two $uW_i$-edges in the linear order for $u$, then every vertex $v$ in $W_i$ is the endpoint of at most one $uW_i$-edge. Indeed, if after the edges have been removed two vertices $u_1,u_2 \in U'$ have a common neighbor $v$ on $W_i$, then at least two $u_1W_i$-edges, say $e_1,e_2$, cross the edge $u_2v$ (or the other way around). As $v$ is closed, $u_2v$ is crossed by an edge $e$ in $G(W_i)$. By fan-planarity (the absence of configurations I and II) $e,e_1$ and $e_2$ have a common endpoint on $W_i$, making $e_1,e_2$ a pair of parallel edges -- a contradiction. So the number of $uW_i$-edges is at most $2|U'|+|W_i| = ||c_i|| + 2 \deg_J(w_i)$.
 \end{claimproof}

 \noindent
 We can now bound the total number of $uW_i$-edges with $i \geq 2$ as follows.
 \begin{align*}
  \sum_{i=2}^k \#\text{$uW_i$-edges } &\leq \sum_{i=2}^k (||c_i|| + 2\deg_J(w_i))\\
   &= \sum_{i=2}^k ||c_i|| + 2|E(J)| - 2\deg_J(w_1)\\
   &\leq \sum_{i=2}^k ||c_i|| + 4(|U|+k)-8 -2|U|\\
   &= \sum_{i=2}^k ||c_i|| + 2|U| + 4(k-2) \leq \sum_{i=2}^k ||c_i|| + ||c_1|| + 4(k-2).
 \end{align*}
\end{proof}

We continue by bounding the total number of crossed edges of $G$ that are drawn inside a fixed face $f$ of $G$. To this end let $k_f$ be the number of distinct facial walks of $f$ and $|f|$ be the sum of lengths of facial walks of $f$, i.e., $|f| = \sum_{W \text{ facial walk of $f$}}|W|$.

\begin{lemma}\label{lem:inside-f}
  The number of edges inside $f$ is at most
  \[
   2|f| + 5(k_f-2) - \sum_{c \in C_\emptyset(f)}\max\{0,||c||-5\}.
  \]
\end{lemma}
\begin{proof}
 We do induction on $k_f$.
 For $k_f=1$, the face $f$ is bounded by a unique facial walk $W$.
 Then by Lemma~\ref{lem:inner-face} there are at most
 \[
  2|W| - 5 - \sum_{c\in C(W)} \max\{0,||c||-5\}
 \]
 edges inside $f$.
 With $|W| = |f|$ and $C_\emptyset(f) = C(W)$ this gives the claimed bound.

 For $k_f \geq 2$, the face $f$ has $k = k_f$ distinct facial walks $W_1,\ldots,W_k$. Let $c$ be an inclusion-minimal cell in $C_*(f)$. Without loss of generality let $W_1$ be the facial walk with $c \in C(W_1)$ and $W_2,\ldots,W_j$ be the facial walks of $f$ that lie inside $c$. In particular we have $2 \leq j \leq k$. Let $G'$ be the graph that is obtained from $G$ after removing all vertices that lie inside $c$. We consider $G'$ with its fan-planar embedding inherited from $G$. Clearly, the face $f'$ in $G'$ corresponding to $f$ in $G$ has exactly $k-(j-1) < k$ facial walks and we have
 \[|f| = |f'| + |W_2| + \cdots + |W_j|.\]
 For $i=1,\ldots,j$ we denote by $c_i$ the cell of $G(W_i)$ containing $W_1$. (Hence $c_1 = c$.) Moreover, let $C = C(W_2)\cup\cdots\cup C(W_j)$. Then
 \[C_\emptyset(f) = (C_\emptyset(f') \cup C) \setminus \{c_1,c_2,\ldots,c_j\}.\]
 Further we partition the edges inside $f$ into three disjoint sets $E_1, E_2, E_3$ as follows:
 \begin{itemize}
  \item The edges in $E_1$ are precisely the edges of $G'$ inside $f'$.
  \item The edges in $E_2$ are precisely the edges of $G$ between $W_1$ and $W_2 \cup \cdots \cup W_j$.
  \item $E_3 = (E(G(W_2))\setminus E(W_2)) \cup \cdots \cup (E(G(W_j))\setminus E(W_j))$.
 \end{itemize}
 Now by induction hypothesis,
 \[
  |E_1| \leq 2|f'| + 5(k-j-1) - \sum_{c \in C_\emptyset(f')}\max\{0,||c||-5\}.
 \]
 By Lemma~\ref{lem:edges-between},
 \[
  |E_2| \leq \sum_{i=1}^j ||c_i|| + 4(j-2) \leq \sum_{i=1}^j\max\{0,||c_i||-5\} + 9j - 8,
 \]
 and by Lemma~\ref{lem:inner-face}, 
 \[
  |E_3| \leq 2(|W_2|+\cdots+|W_j|)-5(j-1)-\sum_{c \in C}\max\{0,||c||-5\}. 
 \]
 Plugging everything together we conclude that the number of edges of $G$ inside $f$ is at most
 \begin{eqnarray*}
  |E_1 \dot{\cup} E_2 \dot{\cup} E_3| &\leq& 2|f'| + 5(k-j-1) - \sum_{c \in C_\emptyset(f')}\max\{0,||c||-5\}\\
  & &+ \sum_{i=1}^j\max\{0,||c_i||-5\} + 9j - 8\\
  & &+ 2(|W_2|+\cdots+|W_j|)-5(j-1)-\sum_{c \in C}\max\{0,||c||-5\}\\
  &=& 2|f| + 5(k-2) -(j-2) -\sum_{c \in C_\emptyset(f)}\max\{0,||c||-5\}\\
  &\leq& 2|f| + 5(k_f-2) -\sum_{c \in C_\emptyset(f)}\max\{0,||c||-5\},
 \end{eqnarray*}
 which concludes the proof.
\end{proof}

Lemma~\ref{lem:inside-f} implies that inside a face $f$ of $H$ there are at most $2|f| + 5(k_f-2)$ edges. Having this, we are now ready to prove our main theorem, namely that every simple fan-planar graph on $n \geq 3$ vertices has at most $5n-10$ edges.
\begin{proof}[Proof of Theorem~\ref{thm:main}]
 Consider a fan-planar graph $G=(V,E)$ on $n$ vertices with properties~\ref{enum:max-uncrossed} and~\ref{enum:maximal}. Let $H$ be the spanning subgraph of $G$ on all uncrossed edges. Note that
 $V(H) = V(G)$.
 We denote by $F(H)$ the set of all faces of $H$. Since every edge $e \in E(H)$ appears either exactly once in two distinct facial walks or exactly twice in the same facial walk, we have
 \[
  \sum_{f \in F(H)} |f| = 2|E(H)|.
 \]

 Further we denote by $k_f$ the number of facial walks for a given face $f$, and by $CC(H)$ the number of connected components of $H$. Since a face with $k$ facial walks gives rise to $k$ connected components of $H$, we have
 \[
  \sum_{f \in F(H)} (k_f-1) = CC(H) - 1.
 \]
 
 Hence we conclude
 \begin{eqnarray*}
  |E(G)| &\overset{\text{Lemma}~\ref{lem:inside-f}}{\leq}& |E(H)| + \sum_{f \in F(H)} (2|f|+5(k_f-2))\\
   &=& |E(H)| + 2\sum_{f \in F(H)} |f| + 5\sum_{f \in F(H)}(k_f-1) -5|F(H)|\\
   &=& 5|E(H)| + 5 CC(H) - 5|F(H)| - 5 = 5|V(H)|-10,
 \end{eqnarray*}
 where the last equation is Euler's formula for the plane embedded graph $H$. With $|V(H)| = |V(G)| = n$ this concludes the proof.
\end{proof}

\section{Discussion}\label{sec:discussion}

We have shown that every simple $n$-vertex fan-planar graph has at most $5n-10$ edges. We have seen that maximum fan-planar graphs carry an underlying planar structure, possibly enhanced by relatively simple local (non-planar) substructures like stars or $2$-hop edges along the boundary of each face, and combinations of both. Such properties make this class attractive for many algorithms commonly used in graph drawing that also use planar subgraphs as a base, and enhance the drawing by additional non-planar edges.

The new concept of fan-planarity opens a variety of possible research directions.
Of course, if we allow $G$ to have parallel edges or loops, there could be arbitrarily many edges, even if the drawing of $G$ is planar. However, if we only forbid configuration~I, but allow configurations~II and~III, can an $n$-vertex topological graph have more than $5n-10$ edges?

If we consider topological multigraphs with \emph{non-homeomorphic parallel edges} and \emph{non-trivial loops}, does our $5n-10$ bound still hold? Here, two parallel edges are non-homeomorphic (a loop is non-trivial) if the bounded component of the plane described by the edges (the edge) contains at least one vertex.
Note for instance that Euler's formula still holds for plane graphs with non-homeomorphic parallel edges and non-trivial loops, and in this case every face still has length at least $3$. 
Hence such plane multigraphs still have at most $3n-6$ edges. 
We strongly conjecture that our $5n-10$ bound also holds for such fan-planar multigraphs.

If we allow non-simple topological graphs, i.e., allow edges to cross more than once and incident edges to cross, does every $n$-vertex fan-planar graph still have at most $5n-10$ edges? We remark, that if we allow both, non-simple drawings and non-homeomorphic parallel edges, then there are $3$-vertex topological graphs with arbitrarily many edges. Let us simply refer to Figure~\ref{fig:non-simple} for such an example. The idea is to start with an edge $e_1$ from $u$ to $v$, and edge $e_i$ starts clockwise next to $e_{i-1}$ at $u$ goes in parallel with $e_{i-1}$ until $e_{i-1}$ ends at $v$, where $e_i$ goes a little further surrounding $v$ once. No two parallel edges are homeomorphic.

 \begin{figure}[htb]
  \centering
  \subfigure[\label{fig:non-simple}]{
   \includegraphics{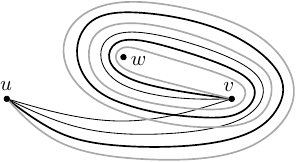}
  }
  \hspace{2em}
    \subfigure[\label{fig:straightlinemax}]{
   \includegraphics[width=4cm]{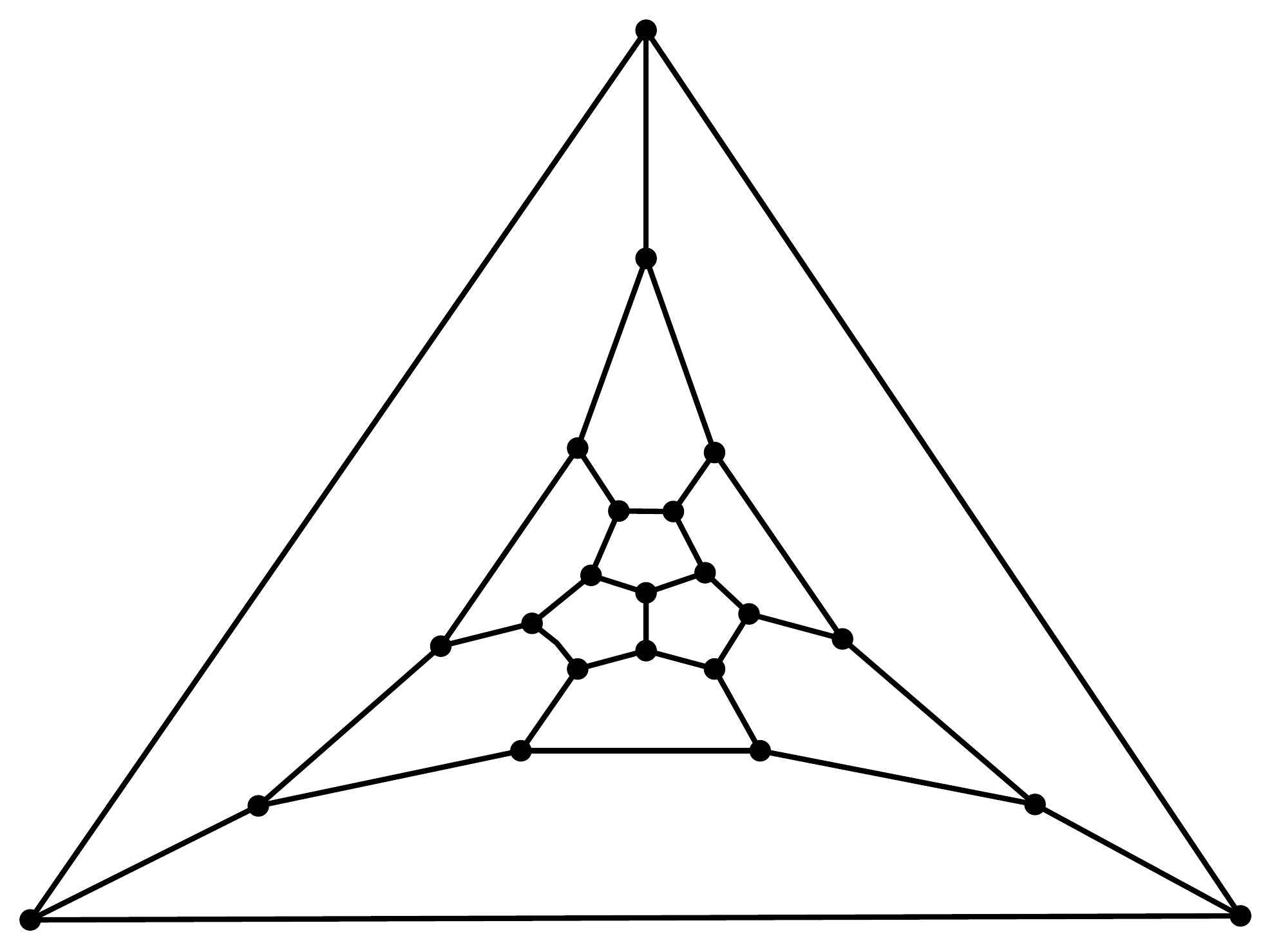}
  } 
   \hspace{2em}
    \subfigure[\label{fig:straightlinemaxext}]{
   \includegraphics[width=4cm]{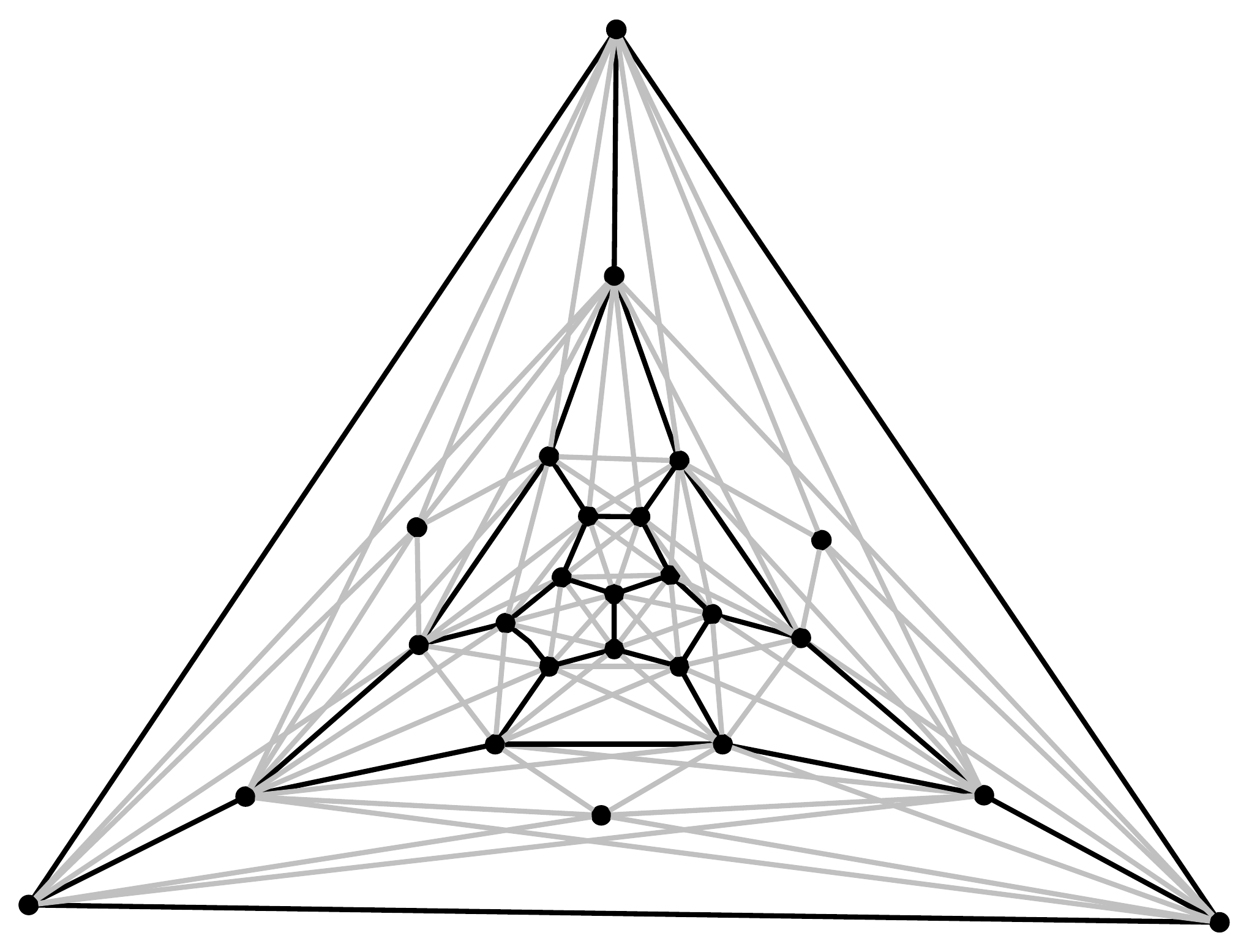}
  }
  \caption{\subref{fig:non-simple} A topological non-simple fan-planar graph with arbitrarily many edges. 
  \subref{fig:straightlinemax} The modified dodecahedral graph without extensions and \subref{fig:straightlinemaxext} fully extended to $5n - 11$ straight-line edges.}
 \end{figure}

If we generalize fan-planarity to \emph{$k$-fan-planarity}, where every edge may be crossed by at most $k$ fan-crossings, $k > 1$, then a simple probabilistic argument from the analysis of natural grids shows that for fixed $k$ every $n$-vertex $k$-fan-planar graph has at most $(3(k+1)^{k+1}/k^k)n$ edges, see Lemma 2.9 in~\cite{DBLP:journals/comgeo/AckermanFPS14}.
However, exact bounds are not known.

Next, we suspect that every $n$-vertex \emph{straight-line} fan-planar graph has at most $5n-11$ edges, similar to the $4n-9$ bound for straight-line $1$-planar graphs~\cite{didimoIPL}. The augmented dodecahedral graph from Figure~\ref{fig:dodecahedral} can be modified into a straight-line fan-planar graph with $5n-11$ edges: Replace one vertex of the dodecahedron by a triangle, which is used as the outer face. Draw the planar graph with convex faces, so that all pentagrams can be drawn straight-line, cf. Figure~\ref{fig:straightlinemax}. The three pentagons that became hexagons are filled with $2$-hops and spokes as explained in Proposition~\ref{prop:LB}, i.e., by one additional vertex and $12$ edges each.

Finally, how \emph{few} edges can an edge-maximal fan-planar graph have? An $n$-cycle with $2$-hop, respectively $3$-hop, edges provides edge-maximal fan-planar graphs with no more than $3n$ edges if parallel edges are allowed and no more than $\frac{8}{3}n$ edges, otherwise. We suspect these examples to be best-possible.

\subsection*{Acknowledgements}

This work started at the Bertinoro Workshop on Graph Drawing 2013.
The authors thank all participants for a wonderful working environment, and some referees for valuable hints especially concerning the relation to the previous work on grids in topological graphs.
 
We are also grateful for the openess and kindness of the authors of~\cite{klemz2021} after noting the issue in our preliminary version, which we took as an incentive to get this paper formally published after all.

\bibliographystyle{plain}
\bibliography{paper}

\end{document}